\documentclass[12pt,a4paper]{amsart}

\usepackage[latin1]{inputenc}     
\usepackage{amsmath,amssymb,amsthm}
\usepackage{latexsym}
\usepackage{amsfonts}
\usepackage{bbm,bm,dsfont} 
\usepackage{color} 
\usepackage{graphicx}

\numberwithin{equation}{section} 

\theoremstyle{definition}
\newtheorem{proposition}{Proposition}

\newtheorem{example}{Example}
\newtheorem{remark}{Remark}
\newtheorem{theorem}{Theorem}
\newtheorem{lemma}{Lemma}

\newcommand{\sfe}{\mathsf{E}}
\newcommand{\sff}{\mathsf{F}}
\newcommand{\sfg}{\mathsf{G}}

\newcommand{\N}{\mathbb N} 
\newcommand{\R}{\mathbb R} 
\newcommand{\C}{\mathbb C} 

\newcommand{\fii}{\varphi} 
\newcommand{\Om}{\Omega} 

\newcommand{\hi}{\mathcal{H}} 
\newcommand{\li}{\mathcal{L}} 
\newcommand{\hd}{{\mathcal{H}_\oplus}} 

\renewcommand{\O}{\mathrm{Obs}} 
\newcommand{\CHI}[1]{\ensuremath{ \chi\raisebox{-1ex}{$\scriptstyle #1$} }}

\newcommand{\lh}{\mathcal{L(H)}} 
\newcommand{\sh}{\mathcal{S(H)}} 

\newcommand{\tr}[1]{\mathrm{tr}\left[#1\right]} 
\def\<{\langle} 
\def\>{\rangle} 
\newcommand{\ket}[1]{|#1\rangle} 
\newcommand{\kb}[2]{|#1 \rangle\langle #2|} 

\newcommand{\Eo}{\mathsf{E}} 
\newcommand{\Fo}{\mathsf{F}} 
\newcommand{\Qo}{\mathsf{Q}} 
\newcommand{\Po}{\mathsf{P}} 
\newcommand{\Mo}{\mathsf{M}} 
\newcommand{\No}{\mathsf{N}} 

\newcommand{\e}{{\bf h}}
\def\d{{\mathrm d}} 
\newcommand{\ov}{\overline} 

\begin{document}

\title{On coexistence and joint measurability of rank-1 quantum observables}

\author{Juha-Pekka Pellonp\"a\"a}
\email{juhpello@utu.fi}
\address{Turku Centre for Quantum Physics, Department of Physics and Astronomy, University of Turku, FI-20014 Turku, Finland}

\begin{abstract}
We show that a rank-1 quantum observable (POVM) $\Mo$ is jointly measurable with a quantum observable $\Mo'$ exactly when
$\Mo'$ is a smearing of $\Mo$.
If $\Mo$ is extreme, rank-1 and discrete then $\Mo$ and $\Mo'$ are coexistent if and only if they are jointly measurable.
\newline

\noindent
PACS numbers: 03.65.Ta, 03.67.--a
\end{abstract}

\maketitle


\thispagestyle{empty}

\section{Introduction}

For observables (POVMs) $\Mo$ and $\Mo'$ of a quantum system, the question arises under which conditions it is possible to collect the experimental data of measurements of $\Mo$ and $\Mo'$ from a measurement of a single observable $\No$ of the same system. For example, if $\Mo$ and $\Mo'$ can be measured together then their measurement outcome statistics can be obtained from the marginals of the joint measurement distribution of their joint
observable $\No$. In this case, $\Mo$ and $\Mo'$ are coexistent, that is, their ranges are contained in the range of $\No$.

To clarify the above definitions, let us consider the case of a discrete POVM $\Mo$ in a finite-dimensional Hilbert space $\hi$ (denote $d=\dim\hi<\infty$).
If $\Om=\{x_1,\,x_2\,\ldots,\,x_N\}$ is the outcome space of $\Mo$ then (by fixing a basis) $\hi\cong\C^d$ and $\Mo$ can be viewed as a collection $(\Mo_1,\,\Mo_2,\ldots,\,\Mo_N)$ of positive $d\times d$--matrices $\Mo_i$ such that $\sum_{i=1}^N\Mo_i=I$ (the identity matrix).
A state of the system is represented as a density matrix $\rho$, that is, a positive matrix of trace 1, and the number
$\tr{\rho\Mo_i}\in[0,1]$ is interpreted as the probability of getting an outcome $x_i$ when a measurement of $\Mo$ is performed and the system is in the state $\rho$.
Actually, $\Mo$ is a map which assigns to each subset $X$ of $\Om$ a positive matrix $\Mo(X)=\sum_{x_i\in X}\Mo_i$ so that
$\tr{\rho\Mo(X)}$ is the probability of getting an outcome belonging to the set $X$. Especially, $\Mo\big(\{x_i\}\big)=\Mo_i$.
Hence, the range of the POVM $\Mo$ (as a map) is the set 
\begin{eqnarray*}
{\rm ran}\,\Mo&=&\big\{\Mo(X)\,|\,X\subset\Om\big\}  \\
&=&\big\{\Mo_1,\,\Mo_2,\Mo_3,\ldots,\,\Mo_1+\Mo_2,\,\Mo_1+\Mo_3,\,\Mo_2+\Mo_3,\ldots,\,\Mo_1+\Mo_2+\Mo_2,\ldots\big\}.
\end{eqnarray*}
Two POVMs $\Mo=(\Mo_i)$ and $\Mo'=(\Mo'_j)$ (of the same Hilbert space) are jointly measurable if there exists a POVM $\No=(\No_{ij})$ such that
$\Mo$ and $\Mo'$ are the marginals of $\No$, that is,
\begin{equation}\label{ghy}
\Mo_i=\sum_j \No_{ij}, \qquad \Mo'_j=\sum_i \No_{ij}.
\end{equation}
Immediately one sees that ${\rm ran}\,\Mo$ is a subset of the range of $\No$. Similarly,
${\rm ran}\,\Mo' \subseteq {\rm ran}\,\No$.
More generally, if the ranges of $\Mo$ and $\Mo'$ belong to the range of some POVM $\No$, i.e.\ ${\rm ran}\,\Mo\cup{\rm ran}\,\Mo'\subseteq {\rm ran}\,\No$
(but equations \eqref{ghy} do not necessarily hold) then $\Mo$ and $\Mo'$ are said to be coexistent.

The notion of the coexistence was introduced by Ludwig \cite{Ludwig} and studied by many authors, see, e.g.\
\cite{Lahti,BuKiLa,LaPu} and references therein. 
Recently Reeb {\it et al}.\ \cite{lyhyt} were able to show that {\it coexistence does not imply joint measurability.} More specifically,
they constructed two observables which are coexistent but cannot be measured together (in the sense of eq.\ \eqref{ghy} above).
None of these observables is rank-1 (a discrete POVM $\Mo$ is rank-1 if every $\Mo_i$ is a rank-1 matrix, i.e.\
the maximum number of linearly independent row (or column) vectors of $\Mo_i$ is 1
or, equivalently, $\Mo_i$ is of the form $\kb{d_i}{d_i}$ where $d_i\in\hi$).

Recall that rank-1 observables have many important properties \cite{HeWo,Part2,Pell,Pell'}. For example, their measurements can be seen as state preparation procedures and the measurements break entanglement completely between the system and its environment.

In this paper, we show that a rank-1 observable $\Mo$ is jointly measurable with an observable $\Mo'$ if and only if $\Mo'$ is a smearing (post-processing) of $\Mo$ (Theorem 1, Remark \ref{remmmm}). For example, a discrete POVM $\Mo'$ is a smearing of a discrete POVM $\Mo$ if there exists a probability (or stochastic or Markov) matrix $(p_{kj})$ such that $0\le p_{kj}\le 1$, $\sum_j p_{kj}=1$, and $\Mo_j'=\sum_k p_{kj}\Mo_k$. 
Furthermore, we show that if $\Mo$ is extreme (i.e.\ an extreme point of the convex set of all observables), discrete and rank-1 then the coexistence of $\Mo$ and $\Mo'$ implies their joint measurability (Theorem 2).
Finally, if the range of a discrete rank-1 observable $\Mo$ is contained in the range of an observable $\ov\Mo$, i.e.\ ${\rm ran}\,\Mo\subseteq {\rm ran}\,\ov\Mo$,
then $\Mo$ and $\ov\Mo$ are jointly measurable (Proposition \ref{proposition1}).

\section{Notations and basic results}

For any Hilbert space $\hi$ we let $\lh$ denote the set of bounded operators on $\hi$. The set of states (density operators i.e.\ positive operators of trace 1) is denoted by $\sh$ and $I_\hi$ is the identity operator of $\hi$.
Throughout this article, we let $\hi$ 
be a {separable}\footnote{That is, $\hi$ has a finite $(d<\infty)$ or countably infinite basis $(d=\infty)$. If $d<\infty$ then $\hi\cong\C^d$ and $\lh$ (resp.\ $\sh$) can be identified with the set of all $d\times d$--complex matrices (resp.\ density matrices).} (complex) nontrivial Hilbert space
and $(\Omega,\Sigma)$ be a measurable space (i.e.\ $\Sigma$ is a $\sigma$-algebra of subsets of a nonempty set $\Omega$).\footnote{Usually in physics, $\Omega$ is finite (or countably infinite) set or a manifold (e.g.\ $\R^n$) when
$\Sigma$ contains e.g.\ open sets (the Borel $\sigma$-algebra of a manifold).} 
In the discrete case, $\Om=\{x_1,\,x_2,\,\ldots\}$ and $\Sigma$ consists of {\it all} subsets of $\Om$.
We denote
$\N:=\{0,1,\ldots\}$ and $\N_\infty:=\N\cup\{\infty\}$.

Let $\O(\Sigma,\,\hi)$ be the  set of {quantum observables}, that is, 
 normalized positive operator (valued) measures (POVMs) $\Mo:\,\Sigma\to\lh$. 
Recall that a map $\Mo:\,\Sigma\to\lh$ is a POVM if and only if 
$X\mapsto\tr{\rho\Mo(X)}$ is a probability measure for all $\rho\in\sh$.
Moreover, $\tr{\rho\Mo(X)}$ is interpreted as the probability of getting an outcome $x$ which belong to $X\in\Sigma$ when a measurement of $\Mo\in\O(\Sigma,\,\hi)$ is performed and the system is in the state $\rho\in\sh$.
The range of a POVM $\Mo:\,\Sigma\to\lh$ is the set
$$
{\rm ran}\,\Mo = \big\{\Mo(X)\,|\,X\in\Sigma\big\}.
$$
Any $\Mo\in\O(\Sigma,\,\hi)$ is called 
a projection valued measure (PVM) or a sharp observable or a spectral measure if $\Mo(X)^2=\Mo(X)$ for all $X\in\Sigma$.

For any observables $\Mo^a,\,\Mo^b\in\O(\Sigma,\,\hi)$ and a number $\lambda\in[0,1]$ one can form the convex combination $\lambda\Mo^a+(1-\lambda)\Mo^b\in\O(\Sigma,\,\hi)$   
which corresponds to classical randomization between the two observables (or their measurements). Hence, the set $\O(\Sigma,\,\hi)$ is convex.
We say that a POVM $\Mo\in\O(\Sigma,\,\hi)$ is {\it extreme} if it is an extreme point of $\O(\Sigma,\,\hi)$, that is, if $\Mo=\lambda\Mo^a+(1-\lambda)\Mo^b$ implies that $\Mo^a=\Mo=\Mo^b$. In other words, extreme POVMs do not allow (nontrivial) convex decompositions and are free from the classical noise arising from this type of randomization.
It is easy to show that PVMs are extreme but there are other extreme POVMs too \cite{Pe11}.
Next we consider minimal Naimark dilations of observables. We start with a simple example of a discrete POVM.

\begin{example}
Assume that $d=\dim\hi<\infty$ and $\Mo=(\Mo_1,\ldots,\,\Mo_N)$ is a discrete POVM. Since each $\Mo_i$ is a positive non-zero $d\times d$--matrix (bounded by the identity matrix), we may write
$$
\Mo_i=\sum_{k=1}^{m_i}\lambda_{ik}\kb{\fii_{ik}}{\fii_{ik}}=\sum_{k=1}^{m_i}\kb{d_{ik}}{d_{ik}}
$$
where the eigenvectors $\fii_{ik}$, $k=1,\ldots,\,m_i$, form an orthonormal set, the eigenvalues $\lambda_{ik}$  are non-zero (and bounded by 1), and $d_{ik}:=\sqrt{\lambda_{ik}}\fii_{ik}$.
We say that $m_i\in\N$ is the multiplicity of the outcome $x_i$ or the rank of $\Mo_i$, and
$\Mo$ is of rank 1 if $m_i=1$ for all $i=1,\ldots,N$. Note that always $m_i\le d$.

Let then $\hd$ be a Hilbert space spanned by an orthonormal basis $e_{ik}$ where $i=1,\ldots,N$ and $k=1,\ldots,m_i$. Obviously, $\dim\hd=\sum_{i=1}^N m_i$.
Define a discrete PVM $\Po=(\Po_1,\ldots,\Po_N)$ of $\hd$ via
$$
\Po_i=\sum_{k=1}^{m_i}\kb{e_{ik}}{e_{ik}}
$$
so that $\hi_{m_i}=\Po_i\hd$ is spanned by the vectors $e_{ik}$, $k=1,\ldots,m_i$,
and we may write (the direct sum)
$$
\hd=\bigoplus_{i=1}^N\hi_{m_i}.
$$
Define a linear operator $J:\,\hi\to\hd$, 
$$
J=\sum_{i=1}^N\sum_{k=1}^{m_i}\kb{e_{ik}}{d_{ik}}
$$
for which
$
J^*\Po_iJ=\Mo_i.
$
Especially, $J^*J=J^*\sum_i\Po_i J=\sum_i\Mo_i=I_\hi$ showing that $J$ is an isometry.
Hence, $\big(\hd,J,\Po\big)$ is a Naimark dilation of $\Mo$. The dilation is minimal, that is,
the span of vectors $\Po_i J\phi$, $i=1,\ldots,N$, $\phi\in\hi$, is the whole $\hd$. Indeed, this follows immediately from equation $\psi=\sum_{i=1}^N\sum_{k=1}^{m_i}\<e_{ik}|\psi\>e_{ik}=
\sum_{i=1}^N\sum_{k=1}^{m_i}\<e_{ik}|\psi\>\lambda_{ik}^{-1}\Po_i J d_{ik}$ where $\psi\in\hd$.
It is well known that $\Mo$ is a PVM if and only if $J$ is unitary (i.e.\ $\{d_{ik}\}_{i,k}$ is an orthonormal basis of $\hi$). In this case one can identify $\hd$ with $\hi$ and $\Po$ with $\Mo$ e.g.\ by setting $e_{ik}=d_{ik}$.

Let then $\{h_n\}_{n=1}^d$ be an orthonormal basis of $\hi$ and define (orthonormal) {\it structure vectors} $\psi_n:=Jh_n\in\hd$ so that $\<e_{ik}|\psi_n\>=\<d_{ik}|h_n\>$ and
\begin{equation}\label{hhiew47}
\Mo(X)=\sum_{x_i\in X}\Mo_i=\sum_{n,m=1}^d\sum_{x_i\in X}\<h_n|\Mo_ih_m\>\kb{h_n}{h_m}
=\sum_{n,m=1}^d\sum_{x_i\in X}\<\psi_n(x_i)|\psi_m(x_i)\>\kb{h_n}{h_m}
\end{equation}
where $\psi_n(x_i)=\Po_i\psi_n=\Po_iJh_n$. Finally, an operator $E\in\li(\hd)$ is decomposable if it commutes with every $\Po_i$, that is, $E=\sum_{i=1}^N E(x_i)$ where $E(x_i):=\Po_iE\Po_i\in\li(\hi_{m_i})$. Now $\Po_i(E\psi)=E(x_i)\Po_i\psi$. 

It is easy to generalize the results of this (discrete) example to the case of an arbitrary (e.g.\ continuous) POVM; just replace the sums
$\sum_{x_i\in X}(\ldots)$ by integrals $\int_X(\ldots)\d\mu(x)$ (see, e.g.\ \cite{Pell'}). This will be done next.
\end{example}

In the rest of this article, we let $\Mo\in\O(\Sigma,\hi)$ be an observable and $\big(\hd,J,\Po\big)$
its minimal (diagonal) Naimark dilation \cite[Theorem 1]{Part1}. Here
$\mu:\,\Sigma\to[0,1]$ is a probability measure\footnote{Or any $\sigma$-finite positive measure on $\Sigma$ such that $\mu$ and $\Mo$ are mutually absolutely continuous, that is, for all $X\in\Sigma$, $\mu(X)=0$ if and only if $\Mo(X)=0$.}
which can always be chosen to be $X\mapsto\tr{\rho_0\Mo(X)}$ where $\rho_0\in\sh$ has only non-zero eigenvalues, and
$$
\hd=\int_\Omega^\oplus\hi_{m(x)}\d\mu(x)
$$
is a direct integral Hilbert space
 with $m(x)$-dimensional fibers (Hilbert spaces) $\hi_{m(x)}$; recall that $\hd$ consists of square integrable `wave functions' $\psi$ such that $\psi(x)\in\hi_{m(x)}$.
 The operator $J:\,\hi\to\hd$ is isometric, and
$$
\Mo(X)=J^*\Po(X) J,\qquad X\in\Sigma,
$$ 
where $\Po\in\O(\Sigma,\hd)$ is the canonical spectral measure (or the `position observable') of $\hd$, that is, for all 
$X\in\Sigma$ and $\psi\in\hd$,
$\Po(X)\psi=\CHI X\psi$ where $\CHI X$ is the characteristic function of $X$.
Moreover, the set of linear combinations of vectors $\Po(X) J\psi$, $X\in\Sigma$, $\psi\in\hi$, is dense in $\hd$.

We say that $m(x)\in\N_\infty$, $m(x)\le\dim\hi$,
is the {\it multiplicity} of the measurement outcome $x\in\Omega$ since $x$ can be viewed as a collection of $m(x)$ outcomes $(x,1),\,(x,2),\ldots$ of some `finer' observable (which can distinguish them) \cite{Pell}.
If $m(x)=1$ for ($\mu$-almost)\footnote{Recall that if one says that some condition holds `for $\mu$-almost all $x\in\Omega$' this means that the condition holds for all $x\in\Omega\setminus O$ where $O$ is a $\mu$-null set, i.e.\ a set of $\mu$-measure zero.
} all $x\in\Omega$ then 
$\Mo$ is of {\it rank 1,} that is, the outcomes of $\Mo$ are `nondegenerate' \cite[Section 4]{Part2}.
For any orthonormal (ON) basis $\e=\{h_n\}_{n=1}^{\dim\hi}$ of $\hi$, define {\it structure vectors} $\psi_n:=Jh_n$ of $\Mo$ so that one can write (weakly)
\begin{eqnarray}\label{eq1}
\Mo(X)=
\sum_{n,m=1}^{\dim\hi}\int_{X}\<\psi_n(x)|\psi_m(x)\>\d\mu(x)\kb{h_n}{h_m}.
\end{eqnarray}
(Compare this equation to \eqref{hhiew47} above; indeed, in the discrete case, $\mu$ is just a counting measure\footnote{A counting measure counts the number of the elements of a (sub)set, i.e.\ $\mu(X)=\# X$ (the number of the elements of the set $X$).} or a sum of Dirac deltas so that all integrals reduce to sums \cite{Pell'}.)
If $\Mo$ is rank-1 then the fibers $\hi_{m(x)}$ are just one-dimensional Hilbert spaces so that, without restricting generality, we may assume that $\hi_{m(x)}\equiv\C$ and thus $\hd=L^2(\mu)$, the space of square integrable wave functions $\psi:\,\Omega\to\C$. Now, for example, the inner product
$\<\psi_n(x)|\psi_m(x)\>$ of $\hi_{m(x)}$ in eq.\ \eqref{eq1} is just $\overline{\psi_n(x)}\psi_m(x)$ (the inner product of the 1-dimensional Hilbert space $\C$).

\begin{example}
Consider a Hilbert space $\hi=L^2(\R)$ spanned by the Hermite functions $h_n$, $n\in\N$. Denote briefly $\ket n=h_n$ and let $a=\sum_{n=0}^\infty\sqrt{n+1}\kb n{n+1}$ be the lowering operator. Let $|z\>=e^{-|z|^2/2}\sum_{n=0}^\infty z^n/\sqrt{n!}\,\ket n$,
$z\in\C$, be a coherent state.
The following rank-1 POVMs are physically relevant (see, eq.\ \eqref{eq1}):
\begin{itemize}
\item The spectral measure\footnote{That is, for any (suitable) wave function $\psi:\,\R\to\C$ one has
$(Q\psi)(x)=x\psi(x)$ and $\big(\Qo(X)\psi\big)(x)=\CHI X(x)\psi(x)$, i.e.\ $Q=\int_\R x\,\d\Qo(x)$.} 
$\Qo(X)=\sum_{n,m=0}^{\infty}\int_{X}\ov{h_n(x)}h_m(x)\d x\kb{n}{m}$
of the position operator $Q=2^{-1/2}(a+a^*)$; now $\Omega=\R$, $\d\mu(x)=\d x$ and $\psi_n(x)=h_n(x)$.
\item The spectral measure
$\No(\{n\})=\kb n n$ of the number operator $N=a^*a$; now $\Omega=\N$, $\mu$ is the counting measure (discrete case) and $\psi_n(x)=\delta_{xn}$ (Kronecker delta) where $x\in\N$.
\item The phase space observable (associated with the $Q$-function)
$\sfg(Z)=\int_Z|z\>\<z|\d^2 z/\pi=\sum_{n,m=0}^{\infty}\int_Z{\ov z}^n z^m/\sqrt{n!m!}\cdot\pi^{-1}e^{-|z|^2}\d^2z\kb{n}{m}$; now $\Omega=\C$, $\d\mu(z)=\pi^{-1}e^{-|z|^2}\d^2z$ and $\psi_n(z)=z^n/\sqrt{n!}$.
\item The canonical phase observable $\Phi(X)=\sum_{n,m=0}^{\infty}\int_{X}e^{i(n-m)\theta}(2\pi)^{-1}\d\theta\kb{n}{m}$; now $\Omega=[0,2\pi)$, $\d\mu(\theta)=(2\pi)^{-1}\d\theta$ and $\psi_n(\theta)=e^{-in\theta}$.
\end{itemize}
\end{example}

Recall that $E\in\li(\hd)$ is decomposable if there exists a (measurable) family of operators $E(x)\in\li(\hi_{m(x)})$, $x\in\Omega$, such that $\mu\text{-ess sup}_{x\in\Omega}\|E(x)\|<\infty$ and
$(E\psi)(x)=E(x)\psi(x)$ for all $\psi\in\hd$ and $\mu$-almost all $x\in\Omega$.

\begin{lemma}\label{lemma}
Let $\sfe:\,\Sigma\to\hi$ be a (possibly non-normalized) positive operator measure. Then $\sfe(X)\le\Mo(X)$ for all $X\in\Sigma$ if and only if there exists a (unique) $E\in\li(\hd)$, $0\le E\le I_\hd$, such that 
$[E,\Po(X)]=0$ and 
$\sfe(X)=J^*\Po(X) EJ$ for all $X\in\Sigma$
if and only if  there exists a (unique) $E\in\li(\hd)$, $0\le E\le I_\hd$,
which is decomposable, $E=\int_\Omega^\oplus E(x)\d\mu(x)$, and
$$
\Eo(X)=
\sum_{n,m=1}^{\dim\hi}\int_{X}\<\psi_n(x)|E(x)\psi_m(x)\>\d\mu(x)\kb{h_n}{h_m}
$$
for all $X\in\Sigma$.
\end{lemma}

\begin{proof}
The first part of the lemma follows immediately from
\cite[Proposition 1]{HaHePe13} (or from \cite[Lemma 1]{Part3}) by noting that any POVM (of $\hi$) can be seen as a completely positive map from an Abelian von Neumann algebra to $\lh$. For example, in the case of $\Mo$, the von Neumann algebra is $L^\infty(\mu)$ (the $\mu$-essentially bounded functions $\Omega\to\C$).
Finally, it is well-known that any bounded operator on $\hd$ is decomposable if and only if it commutes with the canonical spectral measure $\Po$ (see, e.g.\ \cite[Proposition 1]{Part1}).
\end{proof}

\begin{remark}\label{rem}\rm
Note that $\Mo$ is a PVM\footnote{Hence, one can identify $\Mo$ with $\Po$, i.e.\ diagonalize $\Mo$. If $\Mo$ is the spectral measure of a self-adjoint operator $S$ then $\Po$ can be found by solving the eigenvalue equation of $S$.
Now $m(x)$ is the usual multiplicity of the eigenvalue $x\in\R$.} if and only if $J$ is unitary \cite[Theorem 1]{Part1}. In this case,
$\sfe(X)\le\Mo(X)$, $X\in\Sigma$, if and only if $\sfe(X)=J^*\Po(X)JJ^*EJ=\Mo(X)E'$
where $E':=J^*EJ=\sfe(\Omega)$, see Lemma \ref{lemma}.
Moreover, $[E',\Mo(X)]=0$, $X\in\Sigma$.
If $\sfe$ is normalized, i.e.\ $E'=I_\hi$, then $\sfe=\Mo$.
\end{remark}

\section{Jointly measurable observables}

Let $(\Omega,\Sigma)$ and $(\Omega',\Sigma')$ be measurable spaces and $\Mo\in\O(\Sigma,\hi)$
and $\Mo'\in\O(\Sigma',\hi)$. If $[\Mo(X),\Mo'(Y)]=0$ for all $X\in\Sigma$, $Y\in\Sigma'$, then
$\Mo$ and $\Mo'$ are said to {\it commute} (with each other).
Denote the product $\sigma$-algebra\footnote{That is, $\Sigma\otimes\Sigma'$ is the smallest $\sigma$-algebra over $\Om\times\Om'$ which contains the `rectangles' $X\times Y$ where $X\in\Sigma$ and $Y\in\Sigma'$.} 
of $\Sigma$ and $\Sigma'$ by $\Sigma\otimes\Sigma'$.
We say that $\Mo$ and $\Mo'$ are {\it jointly measurable} if there exists an $\No\in\O(\Sigma\otimes\Sigma',\hi)$ such that
$\No(X\times\Omega')=\Mo(X)$ for all $X\in\Sigma$
and 
$\No(\Omega\times Y)=\Mo'(Y)$ for all $Y\in\Sigma'$. In this case, $\No$ is called a {\it joint observable} of 
$\Mo$ and $\Mo'$. 
Physically, this means that the measurement outcome probabilities $\tr{\rho\Mo(X)}$ and $\tr{\rho\Mo'(Y)}$
can be obtained from a measurement of a single observable $\No$ in the state $\rho$ which gives all probabilities $\tr{\rho\No(X\times Y)}$.
Recall that jointly measurable observables need not commute (see, Remark \ref{remmmm}).

Let $(\hd,J,\Po)$ be a minimal (diagonal) Naimark dilation of $\Mo$ where, e.g.\
$\hd=\int_\Omega^\oplus\hi_{m(x)}\d\mu(x)$.
Following \cite[Section 4.2]{JePuVi} we say that $f:\,\Omega\times\Sigma'\to\R$ is
a {\it weak Markov kernel (with respect to $\mu$)} if
\begin{itemize}
\item[(i)] $\Omega\ni x\mapsto f(x,Y)\in\R$ is $\mu$-measurable for all $Y\in\Sigma'$,
\item[(ii)] for every $Y\in\Sigma'$, $0\le f(x,Y)\le 1$ for $\mu$-almost all $x\in\Omega$,
\item[(iii)] $f(x,\Omega')=1$ and $f(x,\emptyset)=0$ for $\mu$-almost all $x\in\Omega$,
\item[(iv)] if $\{Y_i\}_{i=1}^\infty\subseteq\Sigma'$ is a disjoint sequence (i.e.\ $Y_i\cap Y_j=\emptyset$, $i\ne j$) then
$
f(x,\cup_i Y_i)=\sum_i f(x,Y_i)
$ 
for $\mu$-almost all $x\in\Omega$.
\end{itemize}
If there exists a weak Markov kernel $f:\,\Omega\times\Sigma'\to\R$ such that
$\Mo'(Y)=\int_\Omega f(x,Y)\d\Mo(x)$ for all $Y\in\Sigma$ then
$\Mo'$ is a {\it smearing} or a {\it post-processing} of $\Mo$, or any measurement of $\Mo'$ is subordinate to a measurement of $\Mo$
\cite{Holevo}.
Note that one can interpret $f(x,Y)$ as a (classical) conditional probability which is the probability of the event $Y$ assuming that $x$ is obtained. If $\Mo'$ is a post-prosessing of $\Mo$ then the measurement outcome probabilities 
$\tr{\rho\Mo'(Y)}=\int_\Omega f(x,Y)\tr{\rho\Mo(\d x)}$, that is, they can be seen as a classical processing (integration) of  the probability distribution $\tr{\rho\Mo(\d x)}$ related to a measurement of $\Mo$ in the state $\rho$.

\begin{example}
In the case of discrete POVMs, $\Omega=\{x_1,\ldots,x_N\}$ and $\Omega'=\{y_1,\ldots,y_{N'}\}$. Moreover, $\Sigma$ (resp.\ $\Sigma'$) consists of all subsets of $\Omega$ (resp.\ $\Omega'$) and $\mu$ is the counting measure of $\Omega$. Let $f:\,\Omega\times\Sigma'\to\R$ be a weak Markov kernel and denote
$p_{kj}=f\big(x_k,\,\{y_j\}\big)\in[0,1]$. Now 
$$
\sum_{j=1}^{N'} p_{kj}=\sum_{y_j\in\Omega'}f\big(x_k,\,\{y_j\}\big)=f\big(x_k,\,\cup_{y_j\in\Omega'}\{y_j\}\big)
=f\big(x_k,\Omega'\big)
=1
$$ 
so that $(p_{kj})$ is a probability (or stochastic or Markov) matrix, see Introduction.
\end{example}

\begin{remark}\label{remmmm}
Let $f:\,\Omega\times\Sigma'\to\R$ be a weak Markov kernel with respect to $\mu$ (associated with $\Mo$).
Then $\Mo_f:\,\Sigma'\to\lh,\;Y\mapsto \int_\Omega f(x,Y)\d\Mo(x)$ is an observable (i.e.\ a smearing of $\Mo$).
Furthermore, $\Mo$ and $\Mo_f$ are jointly measurable, a joint observable $\No\in\O(\Sigma\otimes\Sigma',\hi)$ being defined by
$$
\No(X\times Y)=\int_X f(x,Y)\d\Mo(x), \qquad X\in\Sigma,\;Y\in\Sigma',
$$
if and only if, for each $\rho\in\sh$,
the probability bimeasure 
$$
\Sigma\times\Sigma'\ni (X,Y)\mapsto \int_X f(x,Y)\tr{\rho\Mo(\d x)}\in[0,1]
$$ 
extends to a probability measure on $\Sigma\otimes\Sigma'$.
For this, one needs additional conditions for $f$, or for the measurable spaces \cite[Section 6]{LaYl}.
If $f$ satisfies a slightly stronger condition
\begin{itemize}
\item[(iv)'] 
for each sequence $\{Y_i\}_{i=1}^\infty\subseteq\Sigma'$,
there exists $\mu$-null set $N\subset X$ such that,
for all $x\in\Omega\setminus N$ and for any disjoint subsequence $\{Y_{i_k}\}_{k=1}^\infty\subseteq \{Y_i\}_{i=1}^\infty$, 
$
f(x,\cup_k Y_{i_k})=\sum_k f(x,Y_{i_k}),
$ 
\end{itemize}
then $\Mo$ and $\Mo_f$ are jointly measurable  \cite[Proposition 1]{Tulcea}. Obviously,
if $Y\mapsto f(x,Y)$ is a probability measure for ($\mu$-almost) all $x\in\Omega$ then (iv)' holds (recall that, in this case, $f$ is called a {\it Markov kernel}).
For example, by choosing $\Sigma'=\Sigma$ and
 $f(x,Y)=\CHI Y(x)$ for all $x\in\Omega$ and $Y\in\Sigma$ then $\No(X\times Y)=\Mo(X\cap Y)$ and $\Mo'(Y)=\Mo(Y)$, $X,\,Y\in\Sigma$, that is, any observable is jointly measurable with itself even if it does not commute with itself.
Clearly, in physical applications, we may always assume that (iv)' holds. Hence, we have seen that (classical) post-processing $\Mo\mapsto \Mo_f$ can be viewed as a joint measurement of $\Mo$ and the smeared $\Mo_f$. 
\end{remark}

Suppose then that $\Mo$ and $\Mo'$ are jointly measurable, and let $\No$ be their joint observable.
Then, for all $X\in\Sigma$ and $Y\in\Sigma'$, $\No(X\times Y)\le \No(X\times\Omega')=\Mo(X)=J^*\Po(X)J$ so that,
by Lemma \ref{lemma}, 
$$
\No(X\times Y)=J^*\Po(X)\Fo(Y) J=\sum_{n,m=1}^{\dim\hi}\int_{X}\<\psi_n(x)|\Fo(x,Y)\psi_m(x)\>\d\mu(x)\kb{h_n}{h_m}
$$
where
$\Fo:\,\Sigma'\to\li(\hd)$ is a POVM which commutes with the canonical spectral measure $\Po$, that is, for each $Y\in\Sigma'$, the operator $\Fo(Y)$ is decomposable,
$$
\Fo(Y)=\int_\Omega^\oplus \Fo(x,Y)\d\mu(x).
$$
Note that $\Fo(x,Y)\in\li(\hi_{m(x)})$ can be chosen to be positive for all $Y\in\Sigma'$ and all $x\in\Omega$.
Moreover,
$$
\Mo'(Y)=J^*\Fo(Y)J=\sum_{n,m=1}^{\dim\hi}\int_{\Omega}\<\psi_n(x)|\Fo(x,Y)\psi_m(x)\>\d\mu(x)\kb{h_n}{h_m},\qquad Y\in\Sigma'.
$$
In addition, if $\Mo$ is rank-1 then $m(x)=1$, $\hi_{m(x)}\cong\C$, and $\hd$ is isomorphic to $L^2(\mu)$, the Hilbert space of the $\mu$-square integrable complex functions on $\Omega$.
In this case, $f(x,Y):=\Fo(x,Y)\in[0,1]$ and $\Fo(Y)\in\li\big(L^2(\mu)\big)$ is a multiplicative operator,
that is, $(\Fo(Y)\psi)(x)=f(x,Y)\psi(x)$ for all $\psi\in\hd$ and for $\mu$-almost all $x\in\Omega$.
Indeed, it is easy to check that $f:\,\Omega\times\Sigma'\to\C,\;(x,Y)\mapsto f(x,Y)$, is a weak Markov kernel with respect to $\mu$ and, since
$$
\No(X\times Y)=\int_X f(x,Y)\d\Mo(x),\qquad
\Mo'(Y)=\int_\Omega f(x,Y)\d\Mo(x),\qquad X\in\Sigma,\;Y\in\Sigma',
$$
$\Mo'$ is a smearing of $\Mo$. Hence, we have proved the following theorem:
\begin{theorem}
Let $\Mo$ be a rank-1 observable and $\Mo'$ any observable.
If $\Mo$ and $\Mo'$ are jointly measurable then $\Mo'$ is a smearing of $\Mo$.
\end{theorem}

Note that, in the general case, $\Fo(x,Y)$ is an operator valued `conditional probability' which operates on the `eigenspace' 
$\hi_{m(x)}$ of $x$. One could say that $\Fo(x,Y)$ is a `quantum (weak) Markov kernel' which reduces to a `classical' kernel $f(x,Y)$ when $\Mo$ is rank-1, i.e.\ when $\Mo$ has the `nondegenerate' outcomes $x$.

\begin{remark}\rm
Let $\Mo$ and $\Mo'$ be any jointly measurable POVMs, and let $\No$ be their joint observable as above.
If $\Mo$ is a PVM then $J$ is unitary and 
$$
\No(X\times Y)=J^*\Po(X)JJ^*\Fo(Y) J=\Mo(X)\Mo'(Y)=\Mo'(Y)\Mo(X),\qquad X\in\Sigma,\; Y\in\Sigma',
$$
that is, $\Mo'$ commutes with $\Mo$. If $\Mo$ is also rank-1 then $\Mo'$ commutes with itself \cite[Theorem 4.4]{JePuVi}.
If both $\Mo$ and $\Mo'$ are PVMs then also $\No$ is projection valued and, for all $Y\in\Sigma'$, $\sff(Y)$ is a projection, that is, $\sff(x,Y)$ is a projection of $\hi_{m(x)}$ for $\mu$-almost all $x\in\Omega$.
Hence, if $\Mo$ and $\Mo'$ are jointly measurable PVMs  
and $\Mo$ is rank-1 (i.e.\ $\hi_{m(x)}=\C$) then $\Mo'$ is a smearing of $\Mo$ given by a weak Markov kernel $f$
such that, for all $Y\in\Sigma'$, $f(x,Y)\in\{0,1\}$ for $\mu$-almost all $x\in\Omega$. In this case, the kernel is `sharp' in the sense that each conditional probability $f(x,Y)$ is either 1 or 0.
\end{remark}

\section{Coexistent observables}
Let $\Mo\in\O(\Sigma,\hi)$ and $\Mo'\in\O(\Sigma',\hi)$. We say that 
$\Mo$ and $\Mo'$ are {\it coexistent} if there exists a $\sigma$-algebra $\ov\Sigma$ over a set $\ov\Om$
and an observable $\ov\Mo:\,\ov\Sigma\to\lh$ such that the ranges of $\Mo$ and $\Mo'$ belong the range of $\ov\Mo$, ${\rm ran}\,\Mo\cup{\rm ran}\,\Mo'\subseteq {\rm ran}\,\ov\Mo$, that is, for any $X\in\Sigma$ and $Y\in\Sigma'$ there exists sets
$Z,\,Z'\in\ov\Sigma$ such that $\ov\Mo(Z)=\Mo(X)$ and $\ov\Mo(Z')=\Mo'(Y)$.
Physically, this means that (if the sets $Z$ and $Z'$ are known) one can obtain the probabilities $\tr{\rho\Mo(X)}$ and
$\tr{\rho\Mo'(Y)}$ from a measurement of $\ov\Mo$ in the state $\rho$. However, there does not necessarily exist a simple `rule' or `formula' from which one could find sets $Z$ and $Z'$ corresponding to $X$ and $Y$.
Obviously, if $\Mo$ and $\Mo'$ are jointly measurable\footnote{Now simply $Z=X\times\Om'$ and $Z'=\Om\times Y$.} 
then they are coexistent but the converse does not hold in general, see \cite[Proposition 1]{lyhyt}.

Suppose then that $\Mo$ is discrete and rank-1. Then, without restricting generality,
we may assume that the outcome space $(\Omega,\Sigma)$ of $\Mo$ is such that $\Omega$ is finite or countably infinite, i.e.\ $\Omega=\{x_1,x_2,\ldots\}$, $x_i\ne x_j$, $i\ne j$, and $\Sigma$ consists of all subsets of $\Omega$ (i.e.\ $\Sigma=2^\Omega$).
Moreover, for all $i\in\{1,2\ldots\}$, $i<\#\Omega+1$,
$$
\Mo\big(\{x_i\}\big)=|d_i\>\<d_i|\ne 0
$$
where vectors $d_i\in\hi$, $d_i\ne 0$, are such that (weakly)
$$
\sum_{i=1}^{\#\Omega}|d_i\>\<d_i|=\Mo(\Omega)=I_\hi.
$$
Note that $\Mo$ is a PVM if and only if the vectors $d_i$ constitute an orthonormal basis of $\hi$.
We declare that $x_i\in\Omega$ is equivalent with $x_j\in\Omega$, and denote $x_i\sim x_j$, if
there exists a $p>0$ such that $|d_i\>\<d_i|=p|d_j\>\<d_j|$. Clearly, $\sim$ is an equivalence relation.
Let $$
[x_i]:=\{x_j\in\Omega\,|\,x_j\sim x_i\}\subseteq\Omega
$$
be the equivalence class of $x_i\in\Omega$ so that $\Omega$ is the disjoint union of the equivalence classes $[x_i]$.
Let $\Omega/{\sim}$ be the quotient set of $\Omega$ by $\sim$.
Now, for all $x_i\in\Omega$,
$$
\Mo([x_i])=\sum_{j=1\atop x_j\sim x_i}^{\#\Omega}|d_j\>\<d_j|=p_i|d_i\>\<d_i|
$$
where $p_i>0$.

Let $\Mo'\in\O(\Sigma',\hi)$ be an arbitrary observable and assume that 
$\Mo$ and $\Mo'$ are coexistent, i.e.\  the ranges of $\Mo$ and $\Mo'$ belong the range of some observable $\ov\Mo\in\O(\ov\Sigma,\hi)$.
For any $[x_i]\in\Omega/{\sim}$, let $Z_{[x_i]}\in\ov\Sigma$ be such that
$\ov\Mo(Z_{[x_i]})=\Mo([x_i])$.
If $[x_i]\ne[x_j]$ then $\ov\Mo(Z_{[x_i]}\cap Z_{[x_j]})\le \ov\Mo(Z_{[x_i]})=p_i|d_i\>\<d_i|$ and
$\ov\Mo(Z_{[x_i]}\cap Z_{[x_j]})\le \ov\Mo(Z_{[x_j]})=p_j|d_j\>\<d_j|$
so that $\ov\Mo(Z_{[x_i]}\cap Z_{[x_j]})=\tilde p_i|d_i\>\<d_i|=\tilde p_j|d_j\>\<d_j|$
where $\tilde p_i,\,\tilde p_j\ge 0$.
If, e.g.\ $\tilde p_i\ne 0$, then $|d_i\>\<d_i|=(\tilde p_j/\tilde p_i)|d_j\>\<d_j|$ yielding a contradiction.
Hence, $\ov\Mo(Z_{[x_i]}\cap Z_{[x_j]})=0$ and 
$$
\ov\Mo\Big(\bigcup_{[x_i]\in\Om/\sim}Z_{[x_i]}\Big)=
\sum_{[x_i]\in\Om/\sim}\ov\Mo\big(Z_{[x_i]}\big)=
\sum_{[x_i]\in\Om/\sim}\Mo\big([x_i]\big)=I_\hi
$$
implying that, for all $Z\in\ov\Sigma$,
$$
\ov\Mo(Z)=\ov\Mo\Big(Z\cap\bigcup_{[x_i]\in\Om/\sim}Z_{[x_i]}\Big)=
\sum_{[x_i]\in\Om/\sim}\ov\Mo\big(Z\cap Z_{[x_i]}\big).
$$
Since $\ov\Mo\big(Z\cap Z_{[x_i]}\big)\le\ov\Mo\big(Z_{[x_i]}\big)=p_i|d_i\>\<d_i|$,
$$
\ov\Mo\big(Z\cap Z_{[x_i]}\big)=p([x_i],Z)p_i|d_i\>\<d_i|=p([x_i],Z)\Mo\big([x_i]\big)
$$
where $p([x_i],Z)\in[0,1]$. Now each mapping $Z\mapsto \ov\Mo\big(Z\cap Z_{[x_i]}\big)$ is $\sigma$-additive and, thus, $Z\mapsto p([x_i],Z)$ is a probability measure for any $[x_i]\in\Omega/{\sim}$.
Define a mapping $f:\,\Omega\times\ov\Sigma\to[0,1]$, $(x_i,Z)\mapsto f(x_i,Z):=p([x_i],Z)$.
It is easy to check that $f$ is a Markov kernel with respect to the counting measure\footnote{That is, $\#X$ is the number of the elements of $X\subseteq\Om$.} $\#:\,2^\Omega\to[0,\infty]$ and, for all $Z\in\ov\Sigma$,
$$
\ov\Mo(Z)=\sum_{[x_i]\in\Om/\sim}p([x_i],Z)\Mo\big([x_i]\big)
=\sum_{[x_i]\in\Om/\sim}p([x_i],Z)\sum_{j=1\atop x_j\sim x_i}^{\#\Omega}|d_j\>\<d_j|
=\sum_{k=1}^{\#\Omega}f(x_k,Z)|d_k\>\<d_k|
$$
so that $\ov\Mo$ is a smearing of $\Mo$. Hence, we have:
\begin{proposition}\label{proposition1}
If the range of a discrete rank-1 observable $\Mo$ is contained in the range of an observable $\ov\Mo$,
then $\Mo$ and $\ov\Mo$ are jointly measurable.
\end{proposition}
For each $Y\in\Sigma'$, let $Z_Y\in\ov\Sigma$ be such that $\ov\Mo(Z_Y)=\Mo'(Y)$.
Then
\begin{equation}\label{equ}
\Mo'(Y)=\sum_{[x_i]\in\Om/\sim}p([x_i],Z_Y)\Mo\big([x_i]\big)=\sum_{k=1}^{\#\Omega}f'(x_k,Y)|d_k\>\<d_k|
\end{equation}
where $f':\,\Omega\times\Sigma'\to[0,1]$ is defined by $f'(x_i,Y):=p([x_i],Z_Y)$.
To show that $f'$ is a Markov kernel (with respect to the counting measure $\#:\,2^\Omega\to[0,\infty]$) one is left to check the $\sigma$-additivity of the mappings $Y\mapsto p([x_i],Z_Y)$.

Assume that the operators $\Mo([x_i])=p_i|d_i\>\<d_i|$, $[x_i]\in\Omega/{\sim}$, are linearly independent in the sense that,
for any (norm) bounded set of real numbers $r_{[x_i]}$, $[x_i]\in\Omega/{\sim}$, the condition
$\sum_{[x_i]\in\Om/\sim}r_{[x_i]}\Mo\big([x_i]\big)=0$ implies $r_{[x_i]}=0$ for all $[x_i]\in\Omega/{\sim}$.
Then, for any disjoint sequence $\{Y_j\}_{j=1}^\infty\subseteq\Sigma'$,
$$
0=\Mo'\big(\cup_{j=1}^\infty Y_j\big)-\sum_{j=1}^\infty\Mo'(Y_j)=\sum_{[x_i]\in\Om/\sim}\Big\{p([x_i],Z_{\cup_j Y_j})-\sum_{j=1}^\infty p([x_i],Z_{Y_j})\Big\}\Mo\big([x_i]\big)
$$
implying that $f'(x_i,\cup_j Y_j)=p([x_i],Z_{\cup_j Y_j})=\sum_{j=1}^\infty p([x_i],Z_{Y_j})=\sum_{j=1}^\infty f'(x_i,Y_j)$
and $Y\mapsto f'(x_i,Y)$ is a probability measure for all $x_i\in\Omega$.
Since $f'$ is a Markov kernel, $\Mo'$ is a smearing of $\Mo$ showing that $\Mo$ and $\Mo'$ are jointly measurable.

Note that $\Mo^\sim:\,2^{\Omega/\sim}\to\lh,\;X'\mapsto \sum_{[x_i]\in X'}\Mo([x_i])$ is a rank-1 observable, a {\it relabeling} of $\Mo$ \cite{HaHePe12}.
If $\Mo^\sim$ is an extreme point of $\O(2^{\Omega/\sim},\hi)$ then the effects $\Mo([x_i])$ are linearly independent in the above sense \cite{HaHePe12} and
$\Mo$ and $\Mo'$ are jointly measurable. 
Especially, if $\Mo$ is extreme in $\O(2^{\Omega},\hi)$ then the operators $\Mo(\{x_i\})$, $i<\#\Omega+1$, are linearly independent.
Moreover, $[x_i]=\{x_i\}$ for all $i<\#\Omega+1$. Hence, we have:

\begin{theorem}
Let $\Mo$ be an extreme rank-1 discrete observable and $\Mo'$ any observable.
Then $\Mo$ and $\Mo'$ are coexistent if and only if they are jointly measurable if and only if $\Mo'$ is a smearing of $\Mo$.
\end{theorem}

The next examples demonstrate that the extremality requirement (or the linear independence) is needed in the above proof.

\begin{example}
Consider a two-dimensional Hilbert space $\C^2$ and fix an orthonormal basis $\ket0:=(1,0)$ and 
$\ket1:=(0,1)$. Define unitary operators $U_k=\kb00+i^k\kb11$, $k=1,2,3,4$.
Let $|d\>:=\frac12(\ket0+\ket1)$ and $\Mo_k:=U_k|d\>\<d|U_k^*$. 
Now $\{\Mo_k\}_{k=1}^4$ is linearly dependent set (with exactly 3 linearly independent operators) and
$
(2-\epsilon)\Mo_1+\epsilon\Mo_2+(2-\epsilon)\Mo_3+\epsilon\Mo_4=I_{\C^2}
$
for all $\epsilon\in\R$.
Clearly, $\big(\Mo_1,\Mo_2,\Mo_3,\Mo_4\big)$ constitute a rank-1 POVM $\Mo$.
Define a 2-outcome rank-2 POVM
$\Mo'=\big(\Mo'_1,\Mo'_2\big)=\big(\frac12 I_{\C^2},\frac12 I_{\C^2}\big)$
so that the ranges of $\Mo$ and $\Mo'$ belong to the range of $\ov\Mo=\Mo$, that is,
$\Mo$ and $\Mo'$ are coexistent. Now equation  \eqref{equ} reduces to $\Mo'_1=\sum_{k=1}^4 p_{k1}\Mo_k$ and $\Mo'_2=\sum_{k=1}^4 p_{k2}\Mo_k$ where the coefficients $p_{kj}\in[0,1]$ are not unique.
Indeed, one can write
$\Mo'_1=\Mo_1+\Mo_3$ and $\Mo'_2=\Mo_2+\Mo_4$ showing that 
{\it $\Mo$ and $\Mo'$ are jointly measurable,} i.e.\ $\Mo'$ is a smearing of $\Mo$ by a Markov kernel $(p_{kj})$ whose non-zero elements $p_{11}$, $p_{31}$, $p_{22}$ and $p_{42}$ are equal to one.
However, if one writes $\Mo'_1=\Mo_1+\Mo_3$ and $\Mo'_2=\Mo_1+\Mo_3$ then, e.g.,
$p_{11}=1$ and $p_{12}=1$. In this case, $p_{11}+p_{12}=2\ne1$ so that $(p_{kj})$ is not a Markov kernel.
\end{example}

\begin{example}
Let $\hi=\C^3$ and let $\{|1\>,\,|2\>,\,|3\>\}$ be its orthonormal basis.
Define orthonormal unit vectors 
$\psi_1:=\big(|1\>+|2\>+|3\>\big)/\sqrt3$, 
$\psi_2:=\big(|1\>+\alpha|2\>+\alpha^2|3\>\big)/\sqrt3$ 
and
$\psi_3:=\big(|1\>+\alpha^2|2\>+\alpha|3\>\big)/\sqrt3$ where
$\alpha:=\exp(2\pi i/3)$ (so that $1+\alpha+\alpha^2=0$ and $\alpha^3=1$).
Define a 6-outcome rank-1 POVM\footnote{Note that $\Mo$ is not extreme since 
$\Mo_1+\Mo_2+\Mo_3=\frac12I_{\C^3}=\Mo_4+\Mo_5+\Mo_6$. Also, $\Mo'$ is not extreme.} 
$$
\Mo=\big(\Mo_1,\Mo_2,\ldots,\Mo_6\big)=\left(\frac12\kb11,\,\frac12\kb22,\,\frac12\kb33,\,
\frac12\kb{\psi_1}{\psi_1},\,\frac12\kb{\psi_2}{\psi_2},\,\frac12\kb{\psi_3}{\psi_3}
\right)
$$
and a 3-outcome rank-2 POVM
$$
\Mo'=\big(\Mo'_1,\Mo'_2,\Mo'_3\big)=\left(
\frac12\kb22+\frac12\kb33,\,
\frac12\kb11+\frac12\kb33,\,
\frac12\kb11+\frac12\kb22
\right).
$$
Since the ranges of $\Mo$ and $\Mo'$ belong to the range of $\Mo$ ($=\ov\Mo$) the observables $\Mo$ and $\Mo'$ are coexistent.
If $\Mo$ and $\Mo'$ are jointly measurable then $\Mo'$ must be a smearing of $\Mo$ (Theorem 1) and, hence, 
of the form
\begin{equation}\label{djwhfvnds}
\Mo'_j=\sum_{k=1}^6 p_{kj}\Mo_k,\qquad j=1,\,2,\,3,
\end{equation}
where $p_{kj}\in[0,1]$ and $\sum_{j=1}^3 p_{kj}=1$. But
$$
0=\<j|\Mo'_j|j\>=\sum_{k=1}^6 \underbrace{p_{kj}\<j|\Mo_k|j\>}_{\ge\;0},\qquad j=1,\,2,\,3,
$$
implying that $p_{kj}\<j|\Mo_k|j\>\equiv 0$. Since, for all $j=1,2,3$ and $k=4,5,6$,
$\<j|\Mo_k|j\>=1/6$ (and thus $p_{kj}=0$) equation \eqref{djwhfvnds} reduces to
$\Mo'_j=\sum_{k=1}^3 p_{kj}\Mo_k$.
Now $\Mo_1'=\Mo_2+\Mo_3$, $\Mo_2'=\Mo_1+\Mo_3$, $\Mo_3'=\Mo_1+\Mo_2$, and the operators $\Mo_k$, $k=1,2,3$, are linearly independent so that one must have
$p_{11}=0$, $p_{12}=1$ and $p_{13}=1$ yielding a contradiction 
$\sum_{j=1}^3p_{1j}=2\ne 1$.
Hence, {\it $\Mo$ and $\Mo'$ are not jointly measurable.}
\end{example}

\begin{remark}
In literature, there exist two significant classes of observables for which coexistence and joint measurability are known to be equivalent (see, e.g.\ \cite{BuKiLa}):
Let $\Mo\in\O(\Sigma,\hi)$ and $\Mo'\in\O(\Sigma',\hi)$ be coexistent and $\ov\Mo\in\O(\ov\Sigma,\hi)$ be such that the ranges of $\Mo$ and $\Mo'$ belong to the range of $\ov\Mo$. Then $\Mo$ and $\Mo'$ are jointly measurable if\footnote{Moreover, one must assume that the measurable spaces $(\Om,\Sigma)$, $(\Om',\Sigma')$ and $(\ov\Om,\ov\Sigma)$ are regular enough, e.g.\ locally compact metrizable and separable topological spaces equipped with their Borel $\sigma$-algebras.}
\begin{enumerate}
\item $\Mo$ (or $\Mo'$) is projection valued, or
\item $\ov\Mo$ is regular, that is, for any $Z\in\ov\Sigma$ such that $0\ne\ov\Mo(Z)\ne I_\hi$ one has
$\ov\Mo(Z)\not\le\frac12 I_\hi$ and $\ov\Mo(Z)\not\ge\frac12 I_\hi$.
\end{enumerate}
Note that  (2) implies that also $\Mo$ and $\Mo'$ are regular but (2) does not imply (1).
In both cases, one need not assume that $\Mo$ (or $\Mo'$) is discrete or rank-1. However, if $\Mo$ is projection valued then it is automatically extreme. Moreover, any rank-1 effect is of the form $p|d\>\<d|$ where $d\in\hi$ is a unit vector and $p\in(0,1]$. It is regular\footnote{If $\dim\hi>1$.} (respectively, a projection) if and only if $p>\frac12$ (resp.\ $p=1$).

Suppose then that $\{|0\>,\,|1\>\}$ is an orthonormal basis of $\hi=\C^2$ and
$\Mo=\big(\Mo_1,\Mo_2,\Mo_3\big)=\big(\kb00,\,0.1\kb11,\,0.9\kb11\big)$ which is not projection valued or regular. Since $\Mo$ has a projection valued (especially, extreme) relabeling 
$\Mo^\sim=\big(\kb00,\,\kb11\big)$ it follows that $\Mo$ and an {\it arbitrary} $\Mo'$ are jointly measurable if and only if they are coexistent.

\end{remark}

\section{Discussion}

It is shown in \cite{Part2} that any observable $\Mo$ can be maximally refined into a rank-1 observable $\Mo_1$ whose value space `contains' also the multiplicities of the measurement outcomes of $\Mo$. We called a measurement of $\Mo_1$ as a {\it complete} measurement of $\Mo$ since (a) it gives information on the multiplicities of the outcomes, (b) it can be seen as a preparation of a new measurement, and (c) it breaks entanglement between the system and its environment \cite{Pell,Pell'}.
Moreover, $\Mo_1$ can be measured by performing a sequential measurement of $\Mo$
and some discrete `multiplicity' observable \cite{Pell}.

Assume then that we measure any observable $\Mo'$ after a measurement of $\Mo_1$.
Since each sequential measurement can be seen as a joint measurement (of $\Mo_1$ and a disturbed `version' $\Mo''$ of $\Mo'$) \cite{LaYl}, it follows from Theorem 1 that $\Mo''$ is a post-processing of $\Mo_1$.
Moreover, the joint observable associated with the sequential measurement is determined by a weak Markov kernel with respect to $\Mo_1$ and the joint measurement can be interpreted as a processing of data obtained from the first measurement (of $\Mo_1$).
Hence, after a complete measurement there is no need to perform any extra measurements.

\noindent {\bf Acknowledgments.} The author thanks Teiko Heinosaari and Roope Uola for useful discussions and comments on the manuscript. This work was supported by the Academy of Finland grant no 138135.


\end{document}